\documentclass[12pt]{article}
\usepackage{_preamble}
\usepackage[subtle]{savetrees}

\begin{document}

\title{Integrating Diagnostic Checks into Estimation\thanks{We thank Isaiah Andrews, Anna Mikusheva, and Alberto Abadie for their guidance and support, as well as MIT lunch seminar participants for helpful comments.}}
\author{Reca Sarfati\thanks{Department of Economics, Massachusetts Institute of Technology, sarfati@mit.edu.} \and Vod Vilfort\thanks{Department of Economics, Massachusetts Institute of Technology, vod@mit.edu.}}
\date{April 17, 2026}
\maketitle

\begin{abstract}
\noindent Empirical researchers often use diagnostic checks to assess the plausibility of their modeling assumptions, such as testing for covariate balance in RCTs, pre-trends in event studies, or instrument validity in IV designs. While these checks are traditionally treated as external hurdles to estimation, we argue they should be integrated into the estimation process itself. In particular, we propose residualizing one's baseline estimator against the vector of diagnostic check statistics to remove the component of baseline sampling variation explained by the diagnostic checks. This residualized estimator offers researchers a ``free lunch,'' delivering three properties simultaneously: (i) eliminating inference distortions from check-based selective reporting; (ii) reducing variance without changing the estimand when the baseline model is correctly specified; and (iii) minimizing worst-case bias under bounded local misspecification within the class of linear adjustments. We apply our method to the RCT in \cite{kaur2024_financialconcerns} and find that, even in a setting where all balance checks pass comfortably, residualization increases the magnitude of the baseline point estimate and reduces its standard error, equivalent to approximately a 10\% increase in sample size.
\end{abstract}

\section{Introduction}
Empirical researchers often use diagnostic checks to guide their choice of estimator. For example, in randomized controlled trials (RCTs), it is common to check for successful randomization by testing for covariate ``balance,'' which examines differences in pre-treatment characteristics across treatment arms \citep{banerjee2017decision, banerjee2020theory}. When such covariates appear imbalanced (e.g., the difference in mean covariates $\hat \gamma = \bar X_1 - \bar X_0$ between the treatment and control groups is sufficiently far from zero), researchers might depart from their initial estimator (e.g., the difference in mean outcomes $\hat{c} = \Bar{Y}_{1} - \Bar{Y}_{0}$ between treatment and control) to adjust outcomes or include controls to address the imbalance. As another example, in difference-in-differences (DiD) and event-study designs, researchers routinely assess the plausibility of parallel trends by examining ``pre-trends'' \citep{roth2022pretest}. Evidence of nonzero pre-trends---for instance, statistically significant coefficients on leads of the treatment indicator---may lead the researcher to modify the specification (e.g., alter the set of controls, change the event window, reweight the sample, etc.). As a further example, in instrumental variables (IV) designs, researchers often report checks of instrument validity, such as balance of predetermined covariates across instrument values or reduced-form effects of the instrument on predetermined outcomes \citep{danieli2026negative}. When such checks yield evidence against the maintained assumptions---for example, when predetermined covariates are imbalanced across instrument values---researchers may similarly modify the specification, such as dropping controls, switching instruments, or restricting the sample to a subgroup where the instrument seems more plausible.

To date, such diagnostic checks have largely been treated as external hurdles to estimation. In this paper, we argue that diagnostic checks should be integrated into the estimation process itself. Specifically,  given a baseline estimator $\hat{c}$ and a vector of jointly asymptotically normal and mean-zero diagnostic checks $\hat{\gamma}$, we propose an estimator $\hat{c}_{R}$ which ``residualizes'' $\hat{\gamma}$ from $\hat{c}$ based on the linear regression of $\hat{c}$ on $\hat{\gamma}$ under their asymptotic normal distribution. In particular, the residualized estimator $\hat{c}_{R}$ takes $\hat{c}$ and subtracts the linear combination of $\hat{\gamma}$ induced by the regression coefficients. Intuitively, the linear adjustment $\hat{c}_{R}$ integrates diagnostic checks into estimation by accounting for the first-order dependence between $\hat{c}$ and $\hat{\gamma}$. 

We show that the residualized estimator $\hat{c}_{R}$ comes with a ``free lunch,'' delivering three properties simultaneously. First, it removes first‑order sensitivity to the check, thereby eliminating bias from selective reporting of the baseline estimator based on passing diagnostic checks. That is, the distribution of $\hat{c}_R$ conditional on ``passing'' a check based on $\hat\gamma$ remains asymptotically normal, with the same mean and variance as its unconditional distribution. In this sense, $\hat{c}_{R}$ is pre-test independent.
Second, under standard regularity conditions, semiparametric efficiency requires orthogonality to any auxiliary statistic whose population value is zero under the base model. For example, in an RCT the covariate imbalance $\hat\gamma=\bar X_1-\bar X_0$ has population value $E[X\mid T=1]-E[X\mid T=0]=0$ under the baseline of random assignment; an efficient average treatment effect (ATE) estimator therefore cannot co-move with $\hat\gamma$ to first order. Otherwise one can strictly reduce variance by residualizing on the check. In particular, residualization is the variance-minimizing linear adjustment of $\hat c$ on $\hat\gamma$; thus, whenever $\hat c$ is not orthogonal to $\hat\gamma$, the linear adjustment $\hat{c}_{R}$ strictly improves large-sample efficiency without changing the estimand under the base model. 
Third, residualization minimizes worst-case bias under bounded local misspecification, which, in combination with our second result, implies that the residualized estimator $\hat{c}_{R}$ is minimax-optimal under squared error loss within the class of linear adjustments.
Below we position these three properties relative to existing literature and results.

\paragraph{Pretesting.}
There are a number of papers that consider the properties of diagnostic checks, with many pointing out issues of power deficiency and inference distortions from pretesting \citep{hartman2018equivalence, freyaldenhoven2019pre, kahn2020promise, roth2022pretest, RambachanRoth2023, bicalho2025power, mikhaeil2025defense, bilinski2026nothing}. We do not address the issue of power, opting instead to take the choice of diagnostic checks as given; efforts to develop better diagnostic checks are complementary to our framework \citep{hartman2018equivalence, bicalho2025power, bilinski2026nothing}. We do, however, consider the issue of pretesting, for which existing solutions include conditional inference corrections and methods that avoid pretesting altogether. Examples of conditional inference applied to RCTs include \cite{andrews2024inference} and \cite{sarfati2025post}, while applications to event studies include \citet{roth2022pretest} and \citet{mikhaeil2025defense}. Examples of approaches that avoid pretesting altogether include \citet{manski2018right}, \citet{freyaldenhoven2019pre}, and \citet{RambachanRoth2023}. Our residualization approach has the benefit of (i) being agnostic to the form of pretesting and (ii) not requiring the specification of bounds, ambiguity sets, or other tuning parameters.\footnote{For example, the conditional inference approach based on truncated Gaussian pivots \citep{fithian2014optimal, lee2016exact, andrews2024inference, sarfati2025post} would require specification of the conditioning set of a given pretesting procedure.} This is useful when there is uncertainty or disagreement over the inputs to (i) and (ii). For example, there can be disagreement over the acceptable threshold for covariate balance \citep{banerjee2020theory}.

\paragraph{Efficiency.} The orthogonality property of efficient estimators is well-established in the literature on semiparametric efficiency \citep{fisher1925theory, bickel1993efficient, newey1994large, van2000asymptotic} and is the basis of  \citet{hausman1978specification} specification tests. Recent work highlighting this property is \citet{borusyak2024revisiting}, who show orthogonality of their efficient event-study imputation estimator to pre-trend coefficients. We show that such orthogonality can be obtained more generally starting from \textit{any} pair of jointly asymptotically normal baseline estimator $\hat{c}$ and diagnostic checks $\hat{\gamma}$ in a wide variety of contexts, including those highlighted in the introduction. In the case of an RCT with $\hat{c} = \Bar{Y}_{1} - \Bar{Y}_{0}$ and $\hat{\gamma} = \Bar{X}_{1} - \Bar{X}_{0}$, the residualized estimator $\hat{c}_{R}$ coincides with the efficient covariate adjustment proposed by \citet{lin2013agnostic}. This type of adjustment is also advanced by \citet{roth2023efficient} for event-study designs with random treatment timing rather than parallel trends. Our framework speaks more generally to any setting where researchers accompany their baseline estimator $\hat{c}$ with auxiliary statistics $\hat{\gamma}$ whose population value should be zero under the base model.

\paragraph{Misspecification.}
Researchers use diagnostic checks to assess deviations from correct model specification, such as failures of randomization in RCTs or violations of parallel trends in event studies. We model the bias from such deviations as a local-to-zero bias term in the researcher's baseline estimator. This follows a large literature on local misspecification \citep{newey1985generalized, conley2012plausibly, andrews2017measuring, andrews2020informativeness, armstrong2021sensitivity, bonhomme2022minimizing, andrews2025purpose}. The local misspecification approach considers bias on the same asymptotic order as sampling variance, which facilitates tractable analysis of bias-variance tradeoffs. We show that in our setting, the worst-case squared bias of a linear adjustment is proportional to its variance, meaning that the variance-minimizing residualization also minimizes worst-case bias. This ``double optimality'' relies on measuring the size of misspecification in terms of a bound on the $L^{2}$ norm of the local misspecification score function, with no further restrictions on misspecification.\footnote{This is without much loss of generality: as shown by \citet{andrews2020informativeness}, bounding the $L^{2}$ norm of the score is asymptotically equivalent to bounding any divergence in the \citet{cressie1984multinomial} family, which includes Kullback--Leibler divergence and Hellinger distance as special cases.} A recent paper that finds a similar double-optimality result when allowing for analogous forms of unrestricted misspecification is \citet{adusumilli2026you}.

The rest of this paper is organized as follows. Section~\ref{sec:setting} formalizes the setting and establishes notation. Section~\ref{sec:implementation} discusses the practical implementation of our procedure. Section~\ref{sec:result1_selection} presents our first result, which is that residualization eliminates the inference distortions from selective reporting based on the passing of diagnostic checks. Section~\ref{sec:result2_efficiency} presents our second result, which is that orthogonality to diagnostic checks is implied by semiparametric efficiency when the base model is correct, which we use to establish the variance reduction from residualization. Section~\ref{sec:result3_minimax} presents our third result, which is that residualization also minimizes worst-case bias under bounded local misspecification within the class of linear adjustments. Section~\ref{sec:empirics} applies our procedure to the example of balance checks for the RCT in \citet{kaur2024_financialconcerns}. Section~\ref{sec:conclusion} concludes. Appendix~\ref{sec:app_a} presents proofs for our results. Appendix~\ref{sec:app_b} provides additional details, remarks, and worked examples.

\section{Setting and Notation}\label{sec:setting}
We consider the following setting, borrowing the terminology of \cite{andrews2025purpose}. A researcher posits an \textit{economic model} $\mathcal M\subseteq\mathbb R\times\Delta(\mathcal D)$, where each element $(c,P)\in\mathcal M$ pairs a scalar parameter of interest $c$ (e.g., a treatment effect) with a distribution $P$ over the observed data $D\in\mathcal D$ on sample space $\mathcal{D}$. 
For instance, in an RCT with treatment $T\in\{0,1\}$ and observed outcome $Y$, the model collects all pairs $(c,P)$ such that $P$ is an observed-data distribution induced by some underlying potential-outcomes structure satisfying random assignment, and $c$ is the corresponding average treatment effect, identified under the model as $c(P)=E_P[Y\mid T=1]-E_P[Y\mid T=0].$
In an event-study design with pre- and post-treatment periods, the model collects all pairs $(c,P)$ such that $P$ is an observed-data distribution induced by some underlying potential-outcomes structure satisfying parallel trends and no anticipation, and $c$ is the corresponding identified average treatment effect on the treated.

The economic model induces a \textit{statistical model} $\mathcal P=\{P: (c,P)\in\mathcal M\text{ for some }c\}\subseteq \Delta(\mathcal D)$, within which we assume the parameter is identified as a functional $c=c(P)$.
The researcher would like to know whether the true (parameter, distribution) pair lies in $\mathcal M$---that is, whether the economic model is correctly specified. In practice, however, the data can only speak to the weaker question of whether $P\in\mathcal P$: the researcher observes $D_1,\ldots,D_n$ drawn i.i.d.\ from some true distribution $P^*$ and can test features of $P^*$, but cannot directly verify that there exists a parameter value $c^*$ such that $(c^*,P^*)\in\mathcal M$. Diagnostic checks target this testable implication. Formally, let $\hat\gamma$ be a $p_\gamma$-vector of statistics with population analogue $\gamma=\gamma(P)\in\mathbb R^{p_\gamma}$ satisfying $\gamma(P)=0$ for all $P\in\mathcal P$. Evidence that $\gamma(P^*)\neq 0$ is thus evidence of \textit{statistical} misspecification ($P^*\notin\mathcal P$), which in turn implies \textit{econometric} misspecification (i.e., there does not exist a parameter value $c^*$ such that $(c^*,P^*)\in\mathcal M$).

In practice, even when $P^{*}\in\mathcal P$, the finite-sample realization $\hat\gamma$ fluctuates around zero, and researchers must make judgment calls about which values of $\hat\gamma$ constitute evidence against the model. We take the choice of diagnostic check as given and ask how to use $\hat\gamma$ constructively---not only as a pass/fail diagnostic, but as an input to the estimator $\hat{c}$ itself.

We now formalize the asymptotic behavior of the estimator and the diagnostic check. For a given $P$, we let $L^{2}(P)$ denote the set of functions $f: \mathcal{D} \to \R$ with finite second moments $E_{P}[f^{2}] \equiv E_{P}[f^{2}(D)] < \infty$, and let $L_{0}^{2}(P) \subseteq L^{2}(P)$ denote the subset with mean zero $E_{P}[f] = 0$. Fix a base distribution $P_0\in\mathcal P$ so that $\hat c$ targets $c=c(P_0)\in\mathbb R$ and $\hat{\gamma}$ targets $\gamma = \gamma(P_{0}) \in \mathbb{R}^{p_{\gamma}}$. To lighten notation, we henceforth write $E_{0} = E_{P_{0}}$. We assume that $(\hat c, \hat\gamma)$ is jointly asymptotically linear: for influence functions $\phi_c \in L_0^2(P_0)$ and $\phi_\gamma\in L_0^2(P_0)^{p_\gamma}$,
\begin{equation}
\sqrt n \begin{pmatrix}
\hat c - c\\
\hat\gamma - \gamma \end{pmatrix}
=
\frac{1}{\sqrt n}\sum_{i=1}^n \begin{pmatrix}
\phi_c(D_i)\\
\phi_\gamma(D_i)
\end{pmatrix} +o_p(1), \quad \Sigma= \begin{bmatrix}
\sigma_c^2 & \Sigma_{c\gamma}\\
\Sigma_{\gamma c} & \Sigma_{\gamma\gamma}
\end{bmatrix}, \quad \Sigma_{\gamma\gamma} \succ 0, \quad \sigma_{c}^{2} > \Sigma_{c\gamma}\Sigma_{\gamma\gamma}^{-1}\Sigma_{\gamma c},
\label{eq:AL_model}
\end{equation}
where $E_0[\phi_c]=0$, $E_0[\phi_\gamma]=0$, and the blocks of the
joint asymptotic covariance $\Sigma$ are
\[
\sigma_c^2=E_0[\phi_c^2],\qquad 
\Sigma_{c\gamma}=E_0[\phi_c\phi_\gamma'],\qquad 
\Sigma_{\gamma\gamma}=E_0[\phi_\gamma\phi_\gamma'].
\] 
In particular, \eqref{eq:AL_model} implies joint asymptotic normality: $\sqrt{n}(\hat{c} - c, (\hat{\gamma} - \gamma)')' \xrightarrow{d} N(0, \Sigma)$. Note that the off-diagonal block $\Sigma_{c\gamma}$ measures the degree to which the baseline estimator $\hat{c}$ co-moves with the diagnostic check $\hat{\gamma}$ to first order, which will be a central component to our analysis.

We illustrate this framework in the context of an RCT balance check, which we will return to as a running example:

\begin{example}[RCT Balance Check]\label{ex:rct_balance}
A researcher conducts an RCT to estimate the effect of a program on an outcome $Y$, with binary treatment assignment $T\in\{0,1\}$ and predetermined covariates $X$. The economic model is
\[
\mathcal M = \left\{
(c,P)\in\mathbb R\times\Delta(\mathcal D) :
c = E_P[Y\mid T=1]-E_P[Y\mid T=0],
\quad E_P[X\mid T=1]=E_P[X\mid T=0] \right\}.
\]
This observed-data restriction is implied by random assignment in an underlying potential-outcomes model, under which $c$ is the average treatment effect. The induced statistical model
\[
\mathcal P=\{P:(c,P)\in\mathcal M \text{ for some } c\}
\]
therefore consists of observed-data distributions under which treatment and control groups are balanced in predetermined covariates, and within this model the parameter is identified as
\[
c(P)=E_P[Y\mid T=1]-E_P[Y\mid T=0].
\]
The researcher estimates $c(P)$ using the difference in means $\hat c=\bar Y_1-\bar Y_0$ and reports the balance statistic $\hat\gamma=\bar X_1-\bar X_0.$ Under any $P\in\mathcal P$,
\[
\gamma(P)=E_P[X\mid T=1]-E_P[X\mid T=0]=0,
\]
so $\hat\gamma$ targets a testable implication of the model.

Now consider possible deviations from $\mathcal M$, focusing on the case where $X$ is baseline income. Suppose, for redistributive reasons, administrators were more likely to assign treatment to worse-off households. This would generate an observed-data distribution $P_1$ under which treated and control households are no longer balanced in $X$, so that
\[
\gamma(P_1)=E_{P_1}[X\mid T=1]-E_{P_1}[X\mid T=0]<0.
\]
Alternatively, administrators may have been receptive to bribes, so that wealthier households were more likely to receive treatment. This would generate a distribution $P_2$ satisfying
\[
\gamma(P_2)=E_{P_2}[X\mid T=1]-E_{P_2}[X\mid T=0]>0.
\]
In either case, the balance check $\hat\gamma$ detects a violation of the observed-data restriction defining $\mathcal P$, so the true distribution lies outside the maintained statistical model. Since $\mathcal P$ is induced by $\mathcal M$, this also means that the maintained \textit{economic} model fails at the true data-generating process (DGP).

The challenge for inference arises even when the model is correctly specified. Under proper randomization ($P^*\in\mathcal P$), the population imbalance satisfies $\gamma(P^*) = 0$, but $\hat\gamma$ fluctuates around zero in finite samples. A common practice is to test for balance by computing a $t$-statistic for the null hypothesis $\gamma = 0$ and compare it to a conventional significance cutoff. If the test rejects---that is, if balance ``fails''---the researcher may change what they report, for example by adding controls, changing specifications, or discarding the RCT results entirely.
\end{example}

The workflow in Example \ref{ex:rct_balance} applies more generally to settings where researchers may wish to correct for the potential bias of a baseline estimator $\hat{c}$ after observing values of diagnostic statistics $\hat{\gamma}$. In the next section, we propose a residualized estimator $\hat{c}_{R}$ that provides a general correction for such settings.

\section{Implementation}\label{sec:implementation}

Based on \eqref{eq:AL_model}, we have the normal approximation
\begin{align*}
    \sqrt{n}
    \begin{pmatrix}
    \hat{c} - c \\
    \hat{\gamma} - \gamma
    \end{pmatrix}
    \sim 
    N(0, \Sigma), \quad \Sigma = 
    \begin{bmatrix}
    \sigma_{c}^2 & \Sigma_{c\gamma} \\
    \Sigma_{\gamma c} & \Sigma_{\gamma\gamma}
\end{bmatrix}.
\end{align*}
The implied linear regression of $\hat{c}$ on $\hat{\gamma}$ is given by the regression coefficient $\Lambda= \Sigma_{c\gamma}\Sigma_{\gamma\gamma}^{-1}$, which measures the sensitivity of the baseline estimator $\hat{c}$ to the diagnostic checks $\hat{\gamma}$. The residualized estimator is defined as the resulting regression residual
\begin{align*}
    \hat{c}_{R} = \hat{c} - \Lambda \hat{\gamma},
\end{align*}
which has a corresponding asymptotic variance
\begin{align*}
    \sigma_{R}^{2} = \sigma_{c}^{2} - \Sigma_{c\gamma}\Sigma_{\gamma\gamma}^{-1}\Sigma_{\gamma c} > 0.
\end{align*}
The residualized estimator $\hat{c}_{R}$ ``partials out'' the component of the baseline estimator $\hat{c}$ that is linearly predictable from the diagnostic check $\hat{\gamma}$. 
In practice, we replace the unknown asymptotic covariance matrix $\Sigma$ with a consistent estimator $\hat\Sigma$ (i.e., $n\hat{\Sigma} \xrightarrow{p} \Sigma$), where $\hat{\sigma}_{c}^{2} > 0$ and $\hat{\Sigma}_{\gamma\gamma}$ is invertible. This leads to feasible plug-in versions of the above quantities. That is, once a researcher can obtain a consistent estimator of the joint asymptotic covariance of $(\hat{c},\hat{\gamma})$---the same ingredient already used to compute standard errors---it is straightforward to implement our procedures. In particular, we can define the plug-in regression coefficient $\hat{\Lambda} = \hat{\Sigma}_{c\gamma}\hat{\Sigma}_{\gamma\gamma}^{-1}$, plug-in residualized estimator $\hat{c}_{R} = \hat{c} - \hat{\Lambda}\hat{\gamma}$, and the plug-in variance estimator $\hat{\sigma}_{R}^{2} = \hat{\sigma}_{c}^{2} - \hat{\Sigma}_{c\gamma}\hat{\Sigma}_{\gamma\gamma}^{-1}\hat{\Sigma}_{\gamma c}$ so that $\hat{\sigma}_{R}$ is the standard error of $\hat{c}_{R}$.

This plug-in approach does not affect the asymptotic distribution of the residualized estimator. To see why, expand $\hat c_R=\hat c-\Lambda\hat\gamma-(\hat\Lambda-\Lambda)\hat\gamma$. The remainder $(\hat\Lambda-\Lambda)\hat\gamma$ is $o_p(n^{-1/2})$ because $\hat\Lambda-\Lambda=o_{p}(1)$ and $\hat\gamma=\gamma+O_p(n^{-1/2})=O_p(n^{-1/2})$ under the base model (where $\gamma=0$). As a result, the plug-in and population-coefficient versions share the same asymptotic distribution to first order. Thus, for ease of exposition, our theoretical results below are stated using the population $\Sigma$ while implementation examples use the sample $\hat{\Sigma}$.

\examplecontinues{ex:rct_balance}
In the RCT example, with $\hat c=\bar Y_1-\bar Y_0$ and $\hat\gamma=\bar X_1-\bar X_0$, residualization gives
\begin{equation}\label{eq:res_est}
    \hat c_R=(\bar Y_1-\bar Y_0)-\hat\Lambda(\bar X_1-\bar X_0),
\end{equation}
where $\hat\Lambda$ is defined as above. In the special case of a linear conditional mean function $E[Y\mid T,X]=\alpha+\tau T+\beta'X$, the population
coefficient satisfies $\Lambda=\beta'$. To see this, write
$Y = \alpha + \tau T + \beta'X + \varepsilon$ with
$E[\varepsilon\mid T, X]=0$. Under independence of $T$ and $X$, we have
\[
\Sigma_{c\gamma}
= \mathrm{Cov}(\bar Y_1-\bar Y_0,\;\bar X_1-\bar X_0)
= \beta'\,\mathrm{Var}(\bar X_1-\bar X_0)
= \beta'\,\Sigma_{\gamma\gamma},
\]
so $\Lambda=\Sigma_{c\gamma}\Sigma_{\gamma\gamma}^{-1}=\beta'$, and the
residualization subtracts $\hat\beta'(\bar X_1-\bar X_0)$ from the
difference-in-means, where $\hat \beta$ is a consistent sample analogue of $\beta$.

\section{Selective Reporting and Pre-Test Independence}
\label{sec:result1_selection}

Empirical practice often conditions what gets reported on whether diagnostic checks ``pass.'' This can take the form of reporting an estimate only when $\hat\gamma$ is close to zero, or modifying the analysis when $|\hat\gamma|$ is large---for instance, switching estimators (changing $\hat c$), or changing the target altogether (changing $c$), e.g., by restricting the sample or redefining the estimand.

To see why this is problematic, let $\tilde c$ denote the reported estimate after the researcher conditions on whether the check ``passes''---for example, $\tilde c = \hat c$ when $\hat\gamma\in\mathcal A$ for some acceptance set $\mathcal A$, and $\tilde c$ is altered or suppressed otherwise. Even if the base model holds, $\tilde c$ is drawn from the \emph{conditional} distribution of $\hat c$ given $\hat\gamma\in\mathcal A$, not its unconditional sampling distribution. 

When $\hat c$ and $\hat\gamma$ co-move, conditioning shifts the distribution of $\tilde c$ relative to its unconditional distribution. Under the null of correct specification, \citet{dechaisemartin2024pretests} show that for symmetric pre-tests (e.g., two-sided $t$-tests of $\gamma=0$), this distortion takes the form of conservatism---confidence intervals over-cover and tests are undersized. More generally, when the base model does not hold, conditioning can distort inference in either direction, so that confidence intervals and $t$-tests based on the unconditional distribution of $\hat c$ need not deliver valid coverage or size after selection \citep{fithian2014optimal, andrews2024inference, sarfati2025post}.

Critically, selection distorts first-order inference whenever the reported estimate co-moves with the check (i.e., $\hat{c}$ has nonzero covariance with $\hat\gamma$). Residualization breaks this link. By construction, the residualized estimator is first-order orthogonal to the check, and since the joint limit is normal, orthogonality implies asymptotic independence. Conditioning on events determined by $\hat\gamma$ therefore has no effect on the first-order distribution of $\hat c_R$. Formally:

\begin{proposition}\label{prop:selection}
Assume \eqref{eq:AL_model} and define $\hat c_R=\hat c-\Lambda\hat\gamma$ with
$\Lambda=\Sigma_{c\gamma}\Sigma_{\gamma\gamma}^{-1}$. Then
\[
\sqrt n
\begin{pmatrix}
\hat c_R-c\\
\hat\gamma-\gamma
\end{pmatrix}
\xrightarrow{d}
N\!\left(
0,\;
\begin{bmatrix}
\sigma_R^2 & 0\\
0 & \Sigma_{\gamma\gamma}
\end{bmatrix}
\right),
\qquad
\sigma_R^2=\sigma_c^2-\Sigma_{c\gamma}\Sigma_{\gamma\gamma}^{-1}\Sigma_{\gamma c},
\]
and, in particular, the Gaussian limit implies $\sqrt n(\hat c_R-c)$ is asymptotically independent of
$\sqrt n(\hat\gamma-\gamma)$.
\end{proposition}

The result has a simple interpretation via the Gaussian limit. If $\big(\sqrt n(\hat c-c),\sqrt n(\hat\gamma-\gamma)\big)$ is jointly normal, the conditional mean of $\hat c$ given $\hat\gamma$ is linear:
\[
E\!\left[\sqrt n(\hat c-c)\mid \sqrt n(\hat\gamma-\gamma)\right]
=\Sigma_{c\gamma}\Sigma_{\gamma\gamma}^{-1}\sqrt n(\hat\gamma-\gamma).
\]
The coefficient $\Lambda=\Sigma_{c\gamma}\Sigma_{\gamma\gamma}^{-1}$ is precisely the component of $\hat c$ that is linearly predictable from $\hat\gamma$. Residualization removes this term, so $E[\sqrt n(\hat c_R-c)\mid \sqrt n(\hat\gamma-\gamma)]=0$ and the conditional distribution of $\hat c_R$ coincides with its unconditional distribution. This asymptotic independence implies post-selection validity for a broad class of conventional diagnostic screening rules, as formalized in the following corollary.

\begin{corollary}[Pre-test independence]\label{cor:pretest_independence}
Under the conditions of Proposition~\ref{prop:selection}, let
\[
T_n=\hat\Sigma_{\gamma\gamma}^{-1/2}\hat\gamma,
\]
and suppose the reporting rule $A_n$ is asymptotically equivalent to a fixed threshold rule based on $T_n$, in the sense that there exist a continuous function $q:\mathbb R^{p_\gamma}\to\mathbb R$ and a constant $a\in\mathbb R$ such that
\[
P\!\left(A_n \,\Delta\, \{q(T_n)\le a\}\right)\to 0.
\]
If, for $Z\sim N(0,I_{p_\gamma})$ we have $0<P(q(Z)\le a)<1$ and $P(q(Z)=a)=0,$ then
\[
\sqrt n(\hat c_R-c)\mid A_n \xrightarrow{d} N(0,\sigma_R^2).
\]
If moreover $\sqrt{n}\hat\sigma_R\xrightarrow{p}\sigma_R$, then the usual one-dimensional Gaussian confidence intervals and tests based on $\hat c_R$ remain asymptotically valid conditional on $A_n$.
\end{corollary}

Corollary~\ref{cor:pretest_independence} shows that the residualized estimator is robust to a broad class of conventional diagnostic screening rules. The key point is that the first-order distribution of $\hat c_R$ is asymptotically independent of the first-order distribution of the diagnostic statistic, so once the reporting decision is driven by a fixed function of the standardized check, conditioning on that decision does not alter the estimator's first-order distribution. In this sense, the correction is built into the estimator itself, rather than applied ex post through a rule-specific conditional-inference adjustment.

The assumption that $A_n$ be asymptotically equivalent to a fixed threshold rule is meant to capture the kinds of screening decisions researchers actually make in practice. It allows the realized reporting decision to differ from the idealized threshold rule on a vanishing set of samples, while requiring that, to first order, the decision be determined by a fixed function of the standardized diagnostic statistic. This covers many rules used in practice, including two-sided $t$-tests, joint Wald or $F$-tests, max-statistic cutoffs, and equivalent $p$-value thresholds. For example, a two-sided screening rule based on the $j$th balance statistic corresponds to $q(t)=|t_j|$, while a joint screening rule based on a Wald statistic corresponds to $q(t)=t't$.

The remaining conditions are mild regularity requirements. The condition $0<P(q(Z)\le a)<1$ rules out degenerate cases in which the screening rule passes with probability tending to zero or one, so that conditioning either becomes ill-posed or asymptotically vacuous. The condition $P(q(Z)=a)=0$ rules out knife-edge cases in which the Gaussian limit places mass exactly on the cutoff, ensuring the pass/fail indicator behaves continuously in the limit. For the threshold-based rules used in applications, such as $t$-tests and Wald tests, this condition is automatic.

The practical benefit of the corollary is that valid first-order inference for $\hat c_R$ does not depend on tailoring a correction to one particular diagnostic rule. If one researcher screens with a 5\% balance test, another with a 1\% test, and a referee prefers a joint Wald rule, the same residualized estimator continues to have the same first-order distribution after conditioning on any of these fixed screening rules. Thus the researcher need not commit ex ante to one specific diagnostic cutoff in order for standard inference on $\hat c_R$ to remain asymptotically valid. 

To illustrate, consider a concrete reporting rule in the RCT setting.

\examplecontinues{ex:rct_balance} In the balance-check example, a common reporting rule is to use the unadjusted difference-in-means (i.e., ``short'' regression) $\hat c_S = \bar Y_1 - \bar Y_0$ when a two-sided $t$-test of $\gamma=0$ finds $|\hat\gamma| \leq t$ ``small'' and switch to the covariate-adjusted ``long'' regression $\hat c_L=\hat c_S-\hat\beta'\hat\gamma$ when $|\hat\gamma| \geq t$ is ``large.'' Since both $\hat c_S$ and $\hat c_L$ generally co-move with $\hat\gamma$, their post-selection distributions differ from their unconditional sampling distributions. By Proposition~\ref{prop:selection}, residualization chooses $\Lambda$ so that $\hat c_R$ is first-order uncorrelated with $\hat\gamma$, eliminating the leading selection effect based on reporting threshold $t$. As \citet{banerjee2020theory} have discussed, however, there is no consensus on what \textit{degree} of imbalance warrants concern, meaning that different researchers may have different personal thresholds $t$ (e.g., corresponding to 5\% vs. 1\% levels of significance). Corollary~\ref{cor:pretest_independence} emphasizes that the solution $\hat c_R$ was not dependent on a particular choice of $t$. The math behind the selection distortion and correction for this particular example is worked out explicitly in Appendix~\ref{sec:app_b}.

\section{Efficiency and Orthogonality to Specification Checks}
\label{sec:result2_efficiency}

Our first result showed that residualization eliminates first-order selection distortions induced by reporting rules based on $\hat\gamma$. In this section we show that the same orthogonality to $\hat\gamma$ also arises from a different source: semiparametric efficiency.
The intuition is simple. A diagnostic check is designed to detect whether the realized data look unusual relative to the maintained statistical model. For example, in an RCT, covariate imbalance is informative about whether treatment assignment looks inconsistent with randomization. Such a check is therefore useful for diagnosing the model, but it is not itself part of the signal about the parameter of interest \textit{under} the model. If an estimator systematically co-moves with the realized value of the check, then some of its sampling variation is being driven by the same fluctuations that make the data look more or less consistent with the maintained restrictions. That component is not informative about $c$, and so an efficient estimator should not contain it.

We now make this precise. Fix a base distribution $P_0\in\mathcal P$. A \emph{smooth one-dimensional submodel through $P_0$} is a scalar-indexed family of distributions $\{P_t:t\in(-\varepsilon,\varepsilon)\}\subseteq\mathcal P$ with $P_{t=0}=P_0$ and score function
\[
s_{P_t}=\left.\frac{d}{dt}\log dP_t\right|_{t=0}\in L_0^2(P_0).
\]
The tangent space of $\mathcal P$ at $P_0$ is the closed linear span of all such scores:
\[
\mathcal T_{\mathcal P}=\overline{\mathrm{span}}\{s_{P_t}\}\subseteq L_0^2(P_0),
\]
where the span ranges over all smooth one-dimensional submodels through $P_0$. Its orthocomplement is
\[
\mathcal T_{\mathcal P}^{\perp}
=\left\{
h\in L_0^2(P_0):E_0[h\,s]=0 \text{ for all } s\in\mathcal T_{\mathcal P}
\right\}.
\]
Intuitively, $\mathcal T_{\mathcal P}$ collects local directions that keep us inside the model, while $\mathcal T_{\mathcal P}^{\perp}$ collects directions orthogonal to all such within-model perturbations and therefore point away from the model.

In the RCT example, a direction in $\mathcal T_{\mathcal P}$ perturbs the joint distribution of the data while preserving random assignment to first order. By contrast, a direction in $\mathcal T_{\mathcal P}^{\perp}$ represents a local departure from randomization.

An estimator is \emph{regular} if its first-order behavior is stable under such local perturbations: along every smooth one-dimensional submodel through $P_0$, the limit distribution of $\sqrt n(\hat c-c(P_t))$ under $P_t$ at $t=1/\sqrt n$ does not depend on the particular direction of the submodel. This rules out pathological superefficient estimators that may look unusually precise at isolated points but behave poorly under arbitrarily small perturbations.

Now let $c=c(P)$ be the target parameter. For a within-model direction $s\in\mathcal T_{\mathcal P}$, write
\[
\dot c(s)=\left.\frac{d}{dt}c(P_t)\right|_{t=0}
\]
for the pathwise derivative of the parameter in direction $s$. If $\hat c$ is a regular estimator with influence function $\phi_c$, then
\[
E_0[\phi_c\,s]=\dot c(s)
\qquad\text{for all } s\in\mathcal T_{\mathcal P}.
\]
This says that when we perturb the distribution slightly in a direction that remains inside the model, the estimator must respond to first order in exactly the same way as the parameter itself.

Semiparametric efficiency adds one more requirement: among all regular estimators that track the parameter correctly in this sense, we want the one with the smallest asymptotic variance. Geometrically, this means choosing the shortest such influence function in $L^2_0(P_0)$. The efficient influence function therefore lies in the tangent space itself: any extra component in $\mathcal T_{\mathcal P}^{\perp}$ would increase variance without helping the estimator track $c$ along within-model directions.

Now consider a diagnostic check. Because $\gamma(P)=0$ for all $P\in\mathcal P$, the check is flat along every within-model perturbation. Formally, for every $s\in\mathcal T_{\mathcal P}$,
\[
\dot\gamma(s) =\left.\frac{d}{dt}\gamma(P_t)\right|_{t=0} =E_0[\phi_\gamma\,s] =0.
\]
Thus $\phi_\gamma\in\mathcal T_{\mathcal P}^{\perp}$: the diagnostic check does not respond to perturbations that stay within the model, and only picks up local directions that move the distribution away from the maintained restrictions. By contrast, the efficient influence function for $c$ lies in $\mathcal T_{\mathcal P}$. An efficient estimator therefore cannot co-move with the diagnostic check to first order.

In the RCT of Example~\ref{ex:rct_balance}, tangent directions in $\mathcal T_{\mathcal P}$ preserve random assignment to first order. Along all such directions balance remains zero, so the influence function of $\hat\gamma=\bar X_1-\bar X_0$ lies in $\mathcal T_{\mathcal P}^{\perp}$. The balance table therefore measures local departures from the randomization model, while an efficient ATE estimator uses only variation that is informative about the treatment effect within the randomization model.
With this setup in hand, we may state the following:

\begin{proposition}[Efficiency implies orthogonality to diagnostic checks]\label{prop:orthonecessary}
Assume joint asymptotic linearity as in \eqref{eq:AL_model}. Suppose $\hat c$ is semiparametrically efficient for $c(P_0)$ under the base model. If $\gamma(P)=0$ for all $P\in\mathcal P$, then
\[
\Sigma_{c\gamma}=E_0[\phi_c(D)\phi_\gamma(D)']=0.
\]
\end{proposition}

Proposition~\ref{prop:orthonecessary} says that an efficient estimator cannot co-move with a statistic that only responds to local departures from the maintained model. The natural follow-up is whether this condition is merely diagnostic or actionable. The next result shows it is the latter: when orthogonality fails, residualization provides an explicit repair.

\begin{corollary}[Efficiency repair]\label{cor:repair}
Under \eqref{eq:AL_model}, if $\hat c$ is regular for $c$, $\gamma(P)=0$ for all $P\in\mathcal P$, and $\Sigma_{c\gamma}\neq 0$, then the residualized estimator $\hat c_R=\hat c-\Lambda\hat\gamma$ is regular for $c$ with strictly smaller large-sample variance:
\[
\mathrm{Var}(\hat c_R)
=\sigma_c^2-\Sigma_{c\gamma}\Sigma_{\gamma\gamma}^{-1}\Sigma_{\gamma c}
<\sigma_c^2.
\]
\end{corollary}

The corollary follows directly from \eqref{eq:AL_model}. For any linear adjustment $\hat c(\lambda)=\hat c-\lambda\hat\gamma$, the influence function is $\phi_c-\lambda\phi_\gamma$, so the asymptotic variance is
\[
E_0[(\phi_c-\lambda\phi_\gamma)^2]
=\sigma_c^2-2\lambda\Sigma_{\gamma c}+\lambda\Sigma_{\gamma\gamma}\lambda',
\]
which is indeed minimized at $\lambda=\Lambda=\Sigma_{c\gamma}\Sigma_{\gamma\gamma}^{-1}$.
The intuition is straightforward: if $\Sigma_{c\gamma}\neq 0$, then part of the estimator's sampling variation is aligned with the same off-model direction captured by the diagnostic check. Residualization subtracts that component, and the improvement is strict whenever $\Sigma_{c\gamma}\neq 0$. In practice, this gives a simple diagnostic: if the estimated covariance $\hat\Sigma_{c\gamma}$ is nonnegligible, the original estimator is leaving precision on the table, which the residualized estimator recovers. To see how this plays out concretely, we can compare three estimators in the RCT setting.

\examplecontinues{ex:rct_balance}
It is useful to compare three estimators in the RCT setting. Let $\hat\gamma=\bar X_1-\bar X_0$ denote the balance statistic and let $\hat c_S=\bar Y_1-\bar Y_0$ denote the unadjusted difference-in-means. A conventional covariate-adjusted estimator is the coefficient on $T$ from the ``long'' regression of $Y$ on $T$ and $X$, which can be written as
\[
\hat c_L=\hat c_S-\hat\beta_L'\hat\gamma, 
\]
where $\hat\beta_L$ is the coefficient on $X$.
The residualized estimator instead uses the coefficient implied by the joint asymptotic covariance of $(\hat c_S,\hat\gamma)$:
\[
\hat c_R=\hat c_S-\hat\beta_R'\hat\gamma,
\qquad
\hat\beta_R=\hat\Sigma_{\gamma\gamma}^{-1}\hat\Sigma_{\gamma c_S}.
\]
Thus both $\hat c_L$ and $\hat c_R$ adjust the short estimator by subtracting a multiple of the realized imbalance $\hat\gamma$, but they use different coefficients. In the discussion below, let $\beta_L$ and $\beta_R$ denote the probability limits of $\hat\beta_L$ and $\hat\beta_R$.

\begin{corollary}[Variance comparisons]\label{cor:rct_variance}
$\hat c_R$ has weakly smaller large-sample variance than both $\hat c_S$ and $\hat c_L$, with equality relative to $\hat c_S$ iff $\Sigma_{c_S\gamma}=0$ and equality relative to $\hat c_L$ iff $\beta_L=\beta_R$.
\end{corollary}

The proof, given in Appendix~\ref{sec:app_a}, follows from the fact that $\beta_R$ minimizes the asymptotic variance of $\hat c_S-\beta'\hat\gamma$ over all $\beta$. Any other choice of coefficient, including $\beta_L$, incurs the additional nonnegative quadratic penalty
\[
(\beta_L-\beta_R)'\Sigma_{\gamma\gamma}(\beta_L-\beta_R).
\]

The two equality cases are informative. First, equality with $\hat c_S$ occurs when $\Sigma_{c_S\gamma}=0$, so that the unadjusted difference-in-means does not co-move with the balance statistic to first order. In that case, the realized imbalance in $X$ is asymptotically uninformative about the sampling error in $\hat c_S$; there is nothing to residualize, and $\hat c_R$ reduces to $\hat c_S$ asymptotically.

Second, equality with $\hat c_L$ occurs when $\beta_L=\beta_R$. A sufficient condition is
\[
E[Y\mid T,X]=\alpha+\tau T+\beta'X,
\]
so that the conditional mean is linear in $X$ and the slope on $X$ is the same in both treatment arms (equivalently, there are no $T\times X$ interactions). In that case, the coefficient on $X$ from the long regression targets the same population coefficient that governs how realized imbalance in $X$ translates into sampling error in $\hat c_S$, and so $\beta_L=\beta_R$. More generally, when the conditional mean is nonlinear in $X$ or the slope on $X$ differs across treatment arms, the two coefficients need not coincide. Thus, while the conventional long regression adjusts for covariates, it does not in general choose the adjustment that is optimal for large-sample precision. Residualization closes this gap by choosing the coefficient that minimizes asymptotic variance, rather than the one implied by a particular regression specification. 

\section{Minimax Bias Under Local Misspecification}
\label{sec:result3_minimax}

In Sections~\ref{sec:result1_selection} and~\ref{sec:result2_efficiency}, we fixed a base distribution $P_0\in\mathcal{P}$ and studied estimators targeting the identified value $c(P_0)$. In that framing, the target was viewed as the value of the functional $c(\cdot)$ at the benchmark distribution $P_0$, and we compared alternative estimators that continue to target this same quantity under correct specification.

For the local misspecification analysis in this section, it is convenient to recast the same benchmark in fixed-target form. That is, fix a benchmark pair $(c_0,P_0)\in\mathcal{M}$. Since $P_0\in\mathcal{P}$ and the parameter is identified on $\mathcal{P}$, this benchmark value satisfies
\[
c_0=c(P_0),
\]
so under correct specification this is only a change in framing, not a change in estimand. The reason for the new framing is that once we allow misspecification, the true DGP need not lie in $\mathcal{P}$, and the model-based functional $c(P)$ need not be defined at the true distribution. We therefore keep the target fixed at the benchmark value $c_0$ and study how estimators for $c_0$ behave when the sampling distribution departs locally from $P_0$.

Specifically, the researcher has a baseline estimator $\hat c$ that is consistent for $c_0$ under $P_0$, together with a diagnostic statistic $\hat\gamma$ whose population value is zero at $P_0$. The true DGP, however, may be a sequence $\{P_n\}$ that differs from $P_0$ by an amount of order $1/\sqrt n$. Under such local misspecification, $\sqrt n(\hat c-c_0)$ need no longer be centered at zero: its limiting distribution may acquire a nonzero mean, so $\hat c$ exhibits first-order asymptotic bias for the fixed target $c_0$. At the same time, the diagnostic check need no longer have population value zero; local misspecification may shift $\gamma(P_n)$ as well.

This fixed-target interpretation is the relevant one in our applications. In the RCT example, for instance, the parameter of interest remains the benchmark ATE $c_0$; a slight failure of random assignment perturbs the sampling distribution of the observed data and can bias the difference-in-means estimator for that target, but it does not itself redefine the estimand of interest. More generally, throughout this section, misspecification refers to a local failure of the baseline model for the DGP, not a change in the target parameter itself. 

We model these departures using the standard local misspecification device from semiparametric theory, and allow the true DGP $P_n$ to deviate from the benchmark distribution $P_0$ by an amount of order $1/\sqrt n$. The misspecification is ``local'' in the sense that $P_n$ converges to $P_0$ at the $1/\sqrt n$ rate, meaning the departure is too small to be separated from $P_0$ without $\sqrt n$-scale magnification, yet large enough to affect first-order asymptotic distributions and hence inference.

A natural concern in this setting is that the estimator $\hat c_R=\hat c-\Lambda\hat\gamma$ may reduce variance under correct specification at the expense of making the estimator more sensitive to misspecification. We show that the opposite is true.

Within the class of linear adjustments $\{\hat c-\lambda\hat\gamma\}$, the residualized estimator $\hat c_R=\hat c-\Lambda\hat\gamma$ minimizes worst-case first-order bias over an unrestricted local misspecification neighborhood, for the same coefficient $\Lambda=\Sigma_{c\gamma}\Sigma_{\gamma\gamma}^{-1}$ from Section~\ref{sec:result2_efficiency} that minimizes asymptotic variance under correct specification. Combining these two results, this implies that $\hat c_R$ is asymptotically minimax under squared loss, within the class of linear adjustments, for any fixed bound on the degree of local misspecification. The remainder of this section is organized as follows: we develop the local misspecification framework, derive the bias formula from primitives, and present the main results in the subsections that follow.

To formalize ``small'' departures from the base model, we adopt the standard local misspecification device from semiparametric theory \citep{bickel1993efficient, van2000asymptotic, newey1985generalized, andrews2017measuring, andrews2020informativeness}. The true data-generating process $P_n$ is allowed to deviate from $P_0$ by an amount of order $1/\sqrt{n}$. To first order, the local alternative may be written as\footnote{Formally, \eqref{eq:local_perturb} is justified by differentiability in quadratic mean of a local path $\{P_t^s\}$ through $P_0$ with score $s$: 
\[
\int\!\big(\sqrt{dP_t^s}-\sqrt{dP_0}-\tfrac{t}{2}s\sqrt{dP_0}\big)^2=o(t^2),
\qquad t\to 0.
\]
setting $P_n=P_{1/\sqrt n}^s$. See \citet{van2000asymptotic} for a formal treatment.}
\begin{equation}\label{eq:local_perturb}
dP_n(d) = \left(1 + \frac{s(d)}{\sqrt{n}}\right) dP_0(d),
\end{equation}
where the score function $s \in L_0^2(P_0)$ satisfies $E_0[s(D)]=0$ (so that $P_n$ integrates to one to first order) and $E_0[s(D)^2]\le\mu^2$ for given $\mu<\infty$. The score function $s$ describes the direction and magnitude in which the distribution $P_n$ departs from $P_0$ (i.e., regions of the support where $s(d)>0$ receive more mass under $P_n$, while regions where $s(d)<0$ receive less). The constraint $E_0[s^2]\le\mu^2$ bounds the size of the perturbation. We consider the local misspecification ball
\begin{equation}
\mathcal S(\mu)=\Big\{s\in L_0^2(P_0):\ E_0[s(D)^2]\le \mu^2\Big\}.
\label{eq:unrestricted_ball}
\end{equation}
This neighborhood is \emph{unrestricted} in the sense that $s$ is not constrained to satisfy $E_0[\phi_\gamma s]=0$ or any other moment condition, so the misspecification may shift the check $\gamma$ itself, contrasting with the \emph{restricted} neighborhood of \citet{andrews2020informativeness}, which imposes $E_0[\phi_\gamma s]=0$ to ensure that admissible misspecifications do not shift the population value of the check. Note that ``unrestricted'' is sometimes used in the literature to mean that misspecification may also perturb the parameter of interest in the population (i.e., viewing $c_0$ as a statistical functional $c(\cdot)$ that changes with $P$), meaning $s$ includes both structural and misspecification directions. Recall that in our fixed-target framing, the adversary is free to shift the check, but the target $c_0$ remains fixed, meaning $s$ represents misspecification of the base model $P_0$ only.

We now derive the first-order asymptotic bias of the linear adjustment $\hat c(\lambda) = \hat c - \lambda\hat\gamma$, with influence function $\psi_\lambda = \phi_c - \lambda\phi_\gamma$. Under the base model $P_0$, asymptotic linearity \eqref{eq:AL_model} gives
\[
\sqrt{n}\big(\hat c(\lambda) - c_0\big) = \frac{1}{\sqrt{n}}\sum_{i=1}^n \psi_\lambda(D_i) + o_p(1).
\]
Under $P_0$ the summands have mean zero, so $\hat c(\lambda)$ is asymptotically unbiased. Under the perturbed distribution $P_n$ defined in \eqref{eq:local_perturb}, the mean of each summand shifts by $\frac{1}{\sqrt{n}}E_0[\psi_\lambda \cdot s]$. By Le Cam's third lemma,\footnote{See \citet[Examples 6.5 and 6.7]{van2000asymptotic}.} this yields
\begin{equation}\label{eq:limiting_dist_misspec}
\sqrt{n}\big(\hat c(\lambda) - c_0\big) \xrightarrow{d} N\!\big(E_0[\psi_\lambda \cdot s],\; E_0[\psi_\lambda^2]\big).
\end{equation}
The first-order asymptotic bias of $\hat c(\lambda)$ under misspecification score $s$ is thus
\begin{equation}\label{eq:bias_formula_detail}
\text{bias}(\lambda, s) = E_0[\psi_\lambda(D)\cdot s(D)] = E_0[\phi_c\cdot s] - \lambda\,E_0[\phi_\gamma \cdot s],
\end{equation}
i.e., the $L^2_0(P_0)$ inner product between the estimator's influence function and the misspecification score. The misspecification score $s$ in \eqref{eq:bias_formula_detail} perturbs the distribution of the data while holding the parameter of interest $c$ fixed. In the RCT of Example~\ref{ex:rct_balance}, for instance, the ATE $c$ is a feature of potential outcomes that exists regardless of whether randomization holds; misspecification of the assignment mechanism (i.e., $P_n \neq P_0$) biases the estimator $\hat c = \bar Y_1 - \bar Y_0$ for this fixed target but does not redefine the ATE. 

\begin{proposition}[Unrestricted worst-case bias and minimax residualization]
\label{prop:minimax}
Under \eqref{eq:AL_model}, the unique minimizer of the worst-case bias over $\lambda$ is
\[
\lambda^\star=\Lambda=\Sigma_{c\gamma}\Sigma_{\gamma\gamma}^{-1},
\qquad
\hat c_R=\hat c(\Lambda)=\hat c-\Lambda\hat\gamma,
\]
and the corresponding minimax worst-case bias is
\begin{equation}
\inf_{\lambda}\ \sup_{s\in\mathcal S(\mu)}\big|E_0[\psi_\lambda s]\big|
=
\mu\sqrt{\sigma_c^2-\Sigma_{c\gamma}\Sigma_{\gamma\gamma}^{-1}\Sigma_{\gamma c}}
=
\mu\,\sigma_c\,\sqrt{1-\mathcal I},
\label{eq:unrestricted_minimax_value}
\end{equation}
where $\mathcal I=\Sigma_{c\gamma}\Sigma_{\gamma\gamma}^{-1}\Sigma_{\gamma c}/\sigma_c^2\in[0,1)$.
\end{proposition}

To sketch the proof: since the bias \eqref{eq:bias_formula_detail} is an inner product between $\psi_\lambda$ and $s$, Cauchy--Schwarz implies that worst-case bias is proportional to the $L^2$ length of $\psi_\lambda$, attained when $s$ aligns with $\psi_\lambda$. Residualization shortens the influence function by projecting out the component explained by $\phi_\gamma$, reducing worst-case bias by a factor of $\sqrt{1-\mathcal{I}}$.\footnote{$\mathcal{I}$ is the population $R^2$ from regressing $\hat c$ on $\hat\gamma$ in their joint asymptotic distribution, which \citet{andrews2020informativeness} call the \textit{informativeness} of $\hat\gamma$ for $\hat c$: it measures the share of variation in the baseline estimate that is explained by the diagnostic check. When $\mathcal{I}$ is close to one, most of $\hat c$'s sampling variation is predictable from the check, and residualization achieves a large reduction in worst-case bias. When $\mathcal{I}$ is close to zero, the check carries little information about the estimator, and residualization has little effect.} Since variance also equals $\|\psi_\lambda\|_{L^2}^2$, the same $\Lambda$ minimizes both, giving a stronger result:

\begin{corollary}[Minimax-optimality under squared error loss]\label{cor:minimax_mse}
Under \eqref{eq:AL_model}, for any bound $\mu \geq 0$ on the degree of misspecification,
\[
\inf_{\lambda} \left\{ \sup_{s \in \mathcal{S}(\mu)} \big(E_0[\psi_\lambda\,s]\big)^2 + \|\psi_\lambda\|_{L^2}^2 \right\}
\]
is uniquely attained at $\lambda^\star = \Lambda = \Sigma_{c\gamma}\Sigma_{\gamma\gamma}^{-1}$, with value $(1+\mu^2)(\sigma_c^2 - \Sigma_{c\gamma}\Sigma_{\gamma\gamma}^{-1}\Sigma_{\gamma c})$.
\end{corollary}

That is, there is no bias-variance tradeoff to navigate: $\hat c_R$ is optimal regardless of the bound $\mu$ on misspecification. Note that this double optimality is a feature of the unrestricted misspecification ball $\mathcal{S}(\mu)$.\footnote{See \citet{adusumilli2026you} for an analogous finding in a decision-theoretic framework with unrestricted forms of likelihood misspecification.} Because $\mathcal{S}(\mu)$ is a ball in $L^2(P_0)$, Cauchy--Schwarz gives worst-case bias $\mu\|\psi_\lambda\|_{L^2}$, which is a monotone function of $\|\psi_\lambda\|_{L^2}^2 = \text{Var}(\hat c(\lambda))$. The same $\Lambda$ therefore minimizes both, and the minimax squared loss objective factors as $(1+\mu^2)\|\psi_\lambda\|_{L^2}^2$, eliminating any dependence on $\mu$.

This proportionality would break if the misspecification neighborhood encoded prior knowledge about how the benchmark model might fail. Suppose, for instance, that the researcher believes the true DGP can depart from $P_0$ only along certain margins---say, the assignment mechanism may be slightly misspecified, while the outcome distribution is correctly modeled. Then the local misspecification score $s$ is not allowed to vary in an arbitrary direction in $L_0^2(P_0)$; instead, it must lie in a smaller set of directions $\mathcal V \subset L_0^2(P_0)$ corresponding to those admissible departures. Variance would remain $\|\psi_\lambda\|_{L^2}^2$, but worst-case bias would equal $\mu\|\Pi_{\mathcal V}\psi_\lambda\|_{L^2}$, where $\Pi_{\mathcal V}$ denotes the $L_{0}^2(P_0)$ projection onto $\mathcal V$. In other words, variance depends on the full size of the influence function, while worst-case bias depends only on the component of the influence function aligned with the directions in which misspecification is thought to occur. These are generally different functions of $\lambda$, so the bias-minimizing and variance-minimizing adjustments may differ, the optimal $\lambda$ would depend on $\mu$, and the researcher would face a genuine bias-variance tradeoff. The unrestricted ball avoids this issue by allowing misspecification in any direction in $L_0^2(P_0)$, so reducing the overall length of the influence function is also the unique way to reduce worst-case bias.

\section{Empirical Application}\label{sec:empirics}

We illustrate our method by applying it to the randomized controlled trial in \citet{kaur2024_financialconcerns}, who study whether financial concerns reduce worker productivity. This application is a natural fit for our framework: the paper reports a standard balance table (their Table~I) alongside its main treatment-effect estimate, providing both the estimator $\hat c$ and the diagnostic check $\hat\gamma$ needed to construct the residualized estimator $\hat c_R$. We first summarize the experimental design and main findings, then implement residualization and compare the original and residualized estimates. Finally, we examine the informativeness $\mathcal{I}$ of the balance checks for the ATE estimate, quantifying how much of the estimator's sampling variation is explained by realized covariate imbalance.

\citet{kaur2024_financialconcerns} conduct a field experiment with 408 low-income male workers in rural Odisha, India, who are hired to produce disposable leaf plates for two weeks during the agricultural lean season. Workers are paid piece rates, so productivity directly determines earnings. The authors randomize the timing of wage payments: treated workers receive an interim payment of accrued earnings partway through the contract, while control workers receive all earnings at the end. The main estimate of interest is the effect of this cash infusion on hourly output in the post-payment period, estimated via a difference-in-differences (DiD) specification at the worker-hour level with round-wave fixed effects and controls.\footnote{The authors describe these controls as ``LASSO-selected,'' but the replication code uses a fixed set (perhaps pre-computed). For simplicity, here we take this set of controls as given.} The authors report a productivity increase of 0.109 standard deviations (6.9\%) among treated workers, with effects concentrated among poorer workers (0.204 standard deviations). 

Alongside these estimates, the authors report a balance table (their Table~I, shown here in Table~\ref{tab:balance}) covering baseline demographics, labor market characteristics, wealth measures, financial worries, and pre-treatment productivity. The table reports differences in means across treatment and control for each covariate, along with $p$-values from regressions of each variable on a treatment indicator — the standard format for assessing whether randomization produced comparable groups. As the authors then report, ``The baseline characteristics do not statistically differ between the treatment and control groups overall (Table I, columns (2) and (3)), indicating a successful randomization procedure'' \citet[page 25]{kaur2024_financialconcerns}. In the language of our framework, these covariate differences constitute the diagnostic check $\hat\gamma = \bar X_1 - \bar X_0$, which should be zero in expectation under random assignment ($P_0$) but fluctuates in finite samples.

The main specification, estimated at the worker-hour level using data from the announcement date onward, is
\begin{align*}
y_{irdh}
&= \beta(\text{Cash}_i \times \text{Post-Pay}_{ird})
+ \beta_A(\text{Cash}_i \times \text{Announcement}_{ird}) \\
&\quad + \tau_P\,\text{Post-Pay}_{ird}
+ \tau_A\,\text{Announcement}_{ird}
+ X_{ird}'\eta
+ \delta_r
+ \varepsilon_{irdh},
\end{align*}
where $y_{irdh}$ is the output of worker $i$ in round-wave $r$ on day $d$ in hour $h$, $\text{Cash}_i$ indicates assignment to the interim-payment group, $\text{Post-Pay}_{ird}$ indicates the days after the interim payment was disbursed, $\text{Announcement}_{ird}$ indicates the period after the payment schedule was announced but before cash was disbursed, $\delta_r$ are round-wave (strata) fixed effects, and $X_{ird}$ is a vector of baseline controls.\footnote{Technically, controls selected via the post-double-selection LASSO procedure of \cite{belloniChern2013}. Consistent with the authors' replication code, for simplicity we will take the set of controls as given.} The coefficient of interest is $\beta$, the average treatment effect of the cash infusion on productivity in the post-payment period.

In the notation of Section~\ref{sec:setting}, the ATE estimate $\hat c = \hat\beta$ is the coefficient on $\text{Cash}_i \times \text{Post-Pay}_{ird}$ from the main specification, and the diagnostic check $\hat\gamma = \bar X_1 - \bar X_0$ is the vector of baseline covariate differences reported in the balance table. Under random assignment ($P_0$), the population imbalance satisfies $\gamma(P_0) = E_{0}[X \mid \text{Cash} = 1] - E_{0}[X \mid \text{Cash} = 0] = 0$. The residualized estimator is
\[
\hat c_R = \hat\beta - \hat\Lambda\,\hat\gamma, \qquad \hat\Lambda = \hat\Sigma_{c\gamma}\hat\Sigma_{\gamma\gamma}^{-1},
\]
where $\hat\Sigma_{c\gamma}$ is the estimated covariance between the ATE estimate and the balance statistics, and $\hat\Sigma_{\gamma\gamma}$ is the estimated covariance matrix of the balance statistics. The coefficient $\hat\Lambda$ captures how much the ATE estimate tends to move with realized covariate imbalance: if, say, a sample where treated workers happen to be wealthier also tends to produce a larger $\hat\beta$, then $\hat\Lambda$ will have a nonzero component along the wealth dimension, and residualization will subtract off this predictable component. Implementation requires only a consistent estimate of the joint covariance of $(\hat\beta, \hat\gamma)$, which can be obtained from the same variance estimation procedure used to compute standard errors for the main specification. We find the following:

\begin{table}[htbp]
\centering
\caption{Original vs.\ Residualized Estimates}\label{tab:comparison}
{\def\sym#1{\ifmmode^{#1}\else\(^{#1}\)\fi}
\begin{tabular}{@{\extracolsep{4pt}}l*{2}{>{\centering\arraybackslash}m{3.5cm}}@{}}
\toprule
 & Original & Residualized \\
 & $\hat c$ & $\hat c_R$ \\
\midrule
\addlinespace[3pt]
\multicolumn{3}{l}{\textit{Panel A. Estimates}} \\ \addlinespace[5pt]
Point estimate & 0.1090 & 0.1220 \\
Std.\ error & (0.0465) & (0.0445) \\
$t$-statistic & 2.3438 & 2.7384 \\
$p$-value & 0.0191 & 0.0062 \\
95\% CI & [0.018, 0.200] & [0.035, 0.209] \\
\addlinespace[8pt]
\multicolumn{3}{l}{\textit{Panel B. Diagnostics}} \\ \addlinespace[5pt]
Informativeness ($\mathcal{I}$) & \multicolumn{2}{c}{0.0819} \\
Bias reduction factor ($\sqrt{1-\mathcal{I}}$) & \multicolumn{2}{c}{0.9582} \\
Variance reduction (\%) & \multicolumn{2}{c}{8.19\%} \\
Correction ($\hat\Lambda\hat\gamma$) & \multicolumn{2}{c}{-0.0130} \\
Equiv. increase in sample size & \multicolumn{2}{c}{10\%} \\
\addlinespace[3pt]
\bottomrule
\end{tabular}}

\end{table}

The point estimate increases from 0.109 to 0.122 after residualization. The correction $\hat\Lambda\hat\gamma = -0.013$ is negative, meaning that the realized covariate imbalances, while small individually, were attenuating the treatment effect in the original estimate. In this particular sample, treated workers happened to draw slightly unfavorable baseline characteristics relative to control workers, and the original treatment-effect estimate was absorbing this noise, which the residualized estimator corrects for.

The combined effect of a slightly larger point estimate and a slightly smaller standard error produces a notable increase in statistical significance: the $p$-value falls from 0.019 to 0.006, the $t$-statistic rises from 2.34 to 2.74, and the lower bound of the 95\% confidence interval shifts from 0.018 to 0.035, moving well away from zero. 

The variance reduction is modest but meaningful: the standard error falls from 0.0465 to 0.0445, a reduction of 8.2\% in variance. This is consistent with the low informativeness $\mathcal{I} = 0.082$, indicating that only about 8\% of the sampling variation in $\hat c$ is linearly predictable from realized covariate imbalance. In a well-executed RCT this is reassuring---proper randomization limits the scope for balance checks to predict the estimator's sampling error. In settings with weaker designs, or where covariates are strongly predictive of outcomes (amplifying the effects of any imbalance), one would expect larger gains. That said, even the modest improvement here is equivalent to a 10\% increase in effective sample size---a meaningful gain obtained at zero cost.

Of separate importance, this residualized estimate is immune to selective reporting concerns based on the balance table. As established in Proposition~\ref{prop:selection}, $\hat c_R$ is asymptotically independent of $\hat\gamma$, so its distribution is unaffected by any rule that conditions on the balance statistics---whether that rule is ``report only if all $p$-values in Table~1 exceed 0.05,'' ``add controls if any covariate is significantly imbalanced,'' or any other function of $\hat\gamma$. The original estimate $\hat c$ does not have this property: its nonzero $\hat\Lambda$ confirms that $\Sigma_{c\gamma} \neq 0$, meaning that conditioning on balance outcomes distorts its sampling distribution, even in this well-randomized experiment. Concretely, if a researcher or referee views the original estimate more favorably when the balance table looks clean (i.e., balance statistics are small), the resulting conditional distribution of $\hat c$ differs from its unconditional distribution---as shown in Section~\ref{sec:result1_selection}, the conditional variance shrinks, leading to conservative inference. The magnitude of this distortion is governed by $\Sigma_{c\gamma}$: the stronger the co-movement between the ATE estimate and the balance outcomes, the larger the gap between the conditional and unconditional distributions of $\hat c$. Table~\ref{tab:decomposition} provides a decomposition of which covariates drive this co-movement. In this application, the largest contributions come from baseline measures of labor market activity and financial worries---variables that one would expect to predict productivity, and therefore to generate co-movement between $\hat c$ and $\hat\gamma$ when imbalanced.

\section{Conclusion}\label{sec:conclusion}

Diagnostic checks are ubiquitous in empirical research. Researchers test for covariate balance in RCTs, examine pre-trends in event studies, and assess instrument validity in IV designs. Yet these checks have traditionally been treated as external diagnostics---pass/fail hurdles that inform whether to trust an estimate, but that play no formal role in the estimate itself. This paper argues that diagnostic checks should be integrated into the estimator, and proposes a concrete method for doing so: residualize the estimator on its diagnostic check, subtracting the component of $\hat c$ that is linearly predictable from the realized check $\hat\gamma$. The residualized estimator $\hat c_R = \hat c - \hat\Lambda\hat\gamma$ requires only a consistent estimate of the joint covariance of $(\hat c, \hat\gamma)$, available already from the machinery needed for standard error calculations. Our method delivers three properties simultaneously, offering a ``free lunch'': (i) it eliminates first-order sensitivity to check-based reporting rules (Proposition~\ref{prop:selection}), (ii) recovers precision lost to nuisance variation whenever the base model is correct and the original estimator is not semiparametrically efficient (Proposition~\ref{prop:orthonecessary} and Corollary~\ref{cor:repair}), and (iii) minimizes worst-case bias under local misspecification (Proposition~\ref{prop:minimax}). Together, these imply that $\hat c_R$ is asymptotically minimax optimal under squared loss, within the class of linear adjustments, for any fixed bound on the degree of local misspecification (Corollary~\ref{cor:minimax_mse}). Within that unrestricted local misspecification problem, there is no bias-variance tradeoff to navigate, and the procedure requires no specification of pretesting rules or tuning parameters.
Our application to the RCT in \citet{kaur2024_financialconcerns} illustrates the procedure in practice. Even in a well-randomized experiment where all balance checks pass comfortably, residualization produces a nontrivial correction, shifting the point estimate and reducing standard error.

\clearpage

\setlength{\bibsep}{0pt}
\bibliography{lib}

\clearpage

\appendix\label{sec:app}

\renewcommand{\thetable}{\thesection.\arabic{table}}
\counterwithin{table}{section}

\section*{Appendix}

\section{Proofs}\label{sec:app_a}
\textsc{Proof of Proposition \ref{prop:selection}:}\label{proof:prop:selection}
By \eqref{eq:AL_model}, $\sqrt n(\hat c_R-c)$ has influence function $\phi_c-\Lambda\phi_\gamma$, so the joint influence function for $\big(\hat c_R,\hat\gamma\big)$ is
\[
\begin{pmatrix}
\phi_c-\Lambda\phi_\gamma\\
\phi_\gamma
\end{pmatrix}.
\]
Its covariance block between the two components equals
\[
E_0\!\big[(\phi_c-\Lambda\phi_\gamma)\phi_\gamma'\big]
=\Sigma_{c\gamma}-\Lambda\Sigma_{\gamma\gamma}
=\Sigma_{c\gamma}-\Sigma_{c\gamma}\Sigma_{\gamma\gamma}^{-1}\Sigma_{\gamma\gamma}
=0,
\]
so the limiting covariance is block diagonal. Joint asymptotic normality then yields asymptotic independence. \qed \vspace{1em}

\noindent\textsc{Proof of Corollary~\ref{cor:pretest_independence} (Pretest Independence):}\label{proof:pretest_independence}
Let
\[
X_n=\sqrt n(\hat c_R-c),
\qquad
T_n=\hat\Sigma_{\gamma\gamma}^{-1/2}\hat\gamma,
\qquad
B_n=\{q(T_n)\le a\}.
\]
By Proposition~\ref{prop:selection} and Slutsky's theorem, $(X_n,T_n)\xrightarrow{d}(X,Z)$, where $X\sim N(0,\sigma_R^2),$ $Z\sim N(0,I_{p_\gamma}),$ and $X$ and $Z$ are independent.
We first show that
\[
X_n\mid B_n \xrightarrow{d} N(0,\sigma_R^2).
\]
Let $B=\{q(Z)\le a\}.$ Since $q$ is continuous and $P(q(Z)=a)=0$, the indicator $1\{q(t)\le a\}$ is $Z$-almost surely continuous. Hence, for any bounded continuous function $g:\mathbb R\to\mathbb R$,
\[
E[g(X_n)1_{B_n}]
\to
E[g(X)1_B].
\]
By independence of $X$ and $Z$,
\[
E[g(X)1_B]=E[g(X)]P(B).
\]
Also,
\[
P(B_n)\to P(B)>0
\]
by the continuous mapping theorem and the assumption $0<P(q(Z)\le a)<1$. Therefore,
\[
E[g(X_n)\mid B_n]
=
\frac{E[g(X_n)1_{B_n}]}{P(B_n)}
\to E[g(X)].
\]
Since this holds for every bounded continuous $g$, it follows that
\[
X_n\mid B_n \xrightarrow{d} X\sim N(0,\sigma_R^2).
\]
Now compare the actual reporting rule $A_n$ to the idealized rule $B_n$. For any bounded continuous $g$ with $\|g\|_\infty\le M$,
\begin{align*}
\left|E[g(X_n)\mid A_n]-E[g(X_n)\mid B_n]\right|
&=
\left|
\frac{E[g(X_n)1_{A_n}]}{P(A_n)}
-
\frac{E[g(X_n)1_{B_n}]}{P(B_n)}
\right| \\
&\le
\left|
\frac{E[g(X_n)(1_{A_n}-1_{B_n})]}{P(A_n)}
\right|
+
\left|
E[g(X_n)1_{B_n}]
\right|
\left|
\frac{1}{P(A_n)}-\frac{1}{P(B_n)}
\right| \\
&\le
\frac{M\,P(A_n\Delta B_n)}{P(A_n)}
+
\frac{M\,|P(A_n)-P(B_n)|}{P(A_n)} \\
&\le
\frac{2M\,P(A_n\Delta B_n)}{P(A_n)}.
\end{align*}
Because $ P(A_n\Delta B_n)\to 0$ and $ P(B_n)\to P(B)>0,$ we also have $P(A_n)\to P(B)>0.$ Hence
\[
\left|E[g(X_n)\mid A_n]-E[g(X_n)\mid B_n]\right|\to 0.
\]
Combining this with the convergence of the conditional distribution given $B_n$ yields
\[
E[g(X_n)\mid A_n]\to E[g(X)]
\]
for every bounded continuous $g$, and therefore
\[
X_n\mid A_n \xrightarrow{d} X\sim N(0,\sigma_R^2).
\]
For the final claim, if $\sqrt{n}\hat\sigma_R\xrightarrow{p}\sigma_R$, then Slutsky's theorem implies
\[
\frac{\hat c_R-c}{\hat\sigma_R}\mid A_n \xrightarrow{d} N(0,1),
\]
so the usual one-dimensional Gaussian confidence intervals and tests based on $\hat c_R$ are asymptotically valid conditional on $A_n$.
\qed\vspace{1em}

\noindent\textsc{Proof of Proposition~\ref{prop:orthonecessary}.}\label{proof:orthonecessary}

Because $\gamma(P)=0$ for all $P\in\mathcal P$, each coordinate of the pathwise derivative of $\gamma$ vanishes along every within-model score $s\in\mathcal T_{\mathcal P}$:
\[
\dot\gamma_j(s)=E_0[\phi_{\gamma,j}s]=0,
\qquad j=1,\ldots,p_\gamma,\ \ s\in\mathcal T_{\mathcal P}.
\]
Hence each coordinate $\phi_{\gamma,j}$ lies in $\mathcal T_{\mathcal P}^{\perp}$. If $\hat c$ is semiparametrically efficient for $c(P_0)$ under the base model, then its influence function equals the efficient influence function for $c$, which lies in $\mathcal T_{\mathcal P}$. Therefore, for each $j$,
\[
E_0[\phi_c\,\phi_{\gamma,j}]=0.
\]
Stacking these identities yields
\[
\Sigma_{c\gamma}=E_0[\phi_c(D)\phi_\gamma(D)']=0. \qed
\]
\vspace{1em}

\noindent\textsc{Proof of Corollary~\ref{cor:repair} (Efficiency Repair):}\label{proof:repair}
Let $\psi_R=\phi_c-\Lambda\phi_\gamma$ and $\Lambda=\Sigma_{c\gamma}\Sigma_{\gamma\gamma}^{-1}$. Since $\gamma(P_0)=0$, the residualized estimator satisfies
\[
\sqrt n(\hat c_R-c)
=
\sqrt n(\hat c-c)-\Lambda\sqrt n(\hat\gamma-\gamma)
=
\frac{1}{\sqrt n}\sum_{i=1}^n \psi_R(D_i)+o_p(1)
\]
by \eqref{eq:AL_model}. Thus $\hat c_R$ is asymptotically linear with influence
function $\psi_R$.
It remains to be shown that $\hat c_R$ is regular for the same target $c$. Let $s\in\mathcal T_{\mathcal P}$ be any within-model score. Because $\hat c$ is regular for $c$, its influence function satisfies
\[
E_0[\phi_c\,s]=\dot c(s).
\]
Also, since $\gamma(P)=0$ for all $P\in\mathcal P$, the pathwise derivative of $\gamma$ vanishes along every within-model direction, so
\[
E_0[\phi_\gamma\,s]=\dot\gamma(s)=0.
\]
Hence
\[
E_0[\psi_R\,s]
=E_0[(\phi_c-\Lambda\phi_\gamma)s]=E_0[\phi_c\,s]-\Lambda E_0[\phi_\gamma\,s]=\dot c(s).
\]
Therefore $\psi_R$ is an influence function for $c$, so $\hat c_R$ is regular for the same target. Its large-sample variance is
\begin{align*}
E_0[\psi_R^2] &= E_0[(\phi_c-\Lambda\phi_\gamma)^2] \\
&= \sigma_c^2-2\Lambda\Sigma_{\gamma c}+\Lambda\Sigma_{\gamma\gamma}\Lambda' \\
&= \sigma_c^2 -2\Sigma_{c\gamma}\Sigma_{\gamma\gamma}^{-1}\Sigma_{\gamma c} + \Sigma_{c\gamma}\Sigma_{\gamma\gamma}^{-1}\Sigma_{\gamma\gamma}\Sigma_{\gamma\gamma}^{-1}\Sigma_{\gamma c} \\
&= \sigma_c^2-\Sigma_{c\gamma}\Sigma_{\gamma\gamma}^{-1}\Sigma_{\gamma c}.
\end{align*}
Since $\Sigma_{\gamma\gamma}\succ 0$, we have $\Sigma_{c\gamma}\Sigma_{\gamma\gamma}^{-1}\Sigma_{\gamma c}>0$
whenever $\Sigma_{c\gamma}\neq 0$. Therefore
\[
\mathrm{Var}(\hat c_R)
= \sigma_c^2 -\Sigma_{c\gamma}\Sigma_{\gamma\gamma}^{-1}\Sigma_{\gamma c} <\sigma_c^2.  \qed
\]
\vspace{1em}

\noindent\textsc{Proof of Corollary~\ref{cor:rct_variance}:}\label{proof:rct_invariance} We can show that $\hat c_R$ has weakly smaller large-sample variance than both $\hat c_S$ and $\hat c_L$: 
Let $\Sigma$ denote the joint asymptotic covariance of $(\hat c_S,\hat\gamma)$, with blocks
$\sigma_S^2=\Sigma_{c_Sc_S}$, $\Sigma_{c_S\gamma}$, and $\Sigma_{\gamma\gamma}\succ0$.
By the variance formula for linear adjustments,
\[
\mathrm{Var}(\hat c_R)
=\sigma_S^2-\Sigma_{c_S\gamma}\Sigma_{\gamma\gamma}^{-1}\Sigma_{\gamma c_S}
\le \sigma_S^2=\mathrm{Var}(\hat c_S),
\]
with equality iff $\Sigma_{c_S\gamma}=0$. Next, write $\hat c_L=\hat c_S-\hat\beta_L'\hat\gamma$ and treat $\beta_L$ as fixed at the probability limit of $\hat\beta_L$ (this suffices for first-order variance comparisons). Then
\begin{align*}
\mathrm{Var}(\hat c_L)
&=\mathrm{Var}(\hat c_S-\beta_L'\hat\gamma)\\
&=\sigma_S^2-2\beta_L'\Sigma_{\gamma c_S}+\beta_L'\Sigma_{\gamma\gamma}\beta_L\\
&=\underbrace{\bigl(\sigma_S^2-\Sigma_{c_S\gamma}\Sigma_{\gamma\gamma}^{-1}\Sigma_{\gamma c_S}\bigr)}_{\mathrm{Var}(\hat c_R)}
\;+\;(\beta_L-\beta_R)'\Sigma_{\gamma\gamma}(\beta_L-\beta_R)\\
&\ge \mathrm{Var}(\hat c_R),
\end{align*}
where $\beta_R=\Sigma_{\gamma\gamma}^{-1}\Sigma_{\gamma c_S}$, and the last inequality uses $\Sigma_{\gamma\gamma}\succ0$. Equality holds iff $\beta_L=\beta_R$. \qed \vspace{1em}

\noindent\textsc{Proof of Proposition~\ref{prop:minimax}:}\label{proof:minimax}
We first show that for any $\lambda$,
\begin{equation}
\sup_{s\in\mathcal S(\mu)}\big|E_0[\psi_\lambda s]\big|
=\mu\,\|\psi_\lambda\|_{L^2(P_0)}.
\label{eq:unrestricted_wcbias}
\end{equation}
Fix $\lambda$ and set $\psi_\lambda=\phi_c-\lambda\phi_\gamma$. By the Cauchy--Schwarz inequality, for any $s\in\mathcal S(\mu)$,
\[
|E_0[\psi_\lambda s]|\le \|\psi_\lambda\|_{L^2}\,\|s\|_{L^2}\le \mu\,\|\psi_\lambda\|_{L^2}.
\]
If $\psi_\lambda\neq 0$, equality is attained by $s^\star=\mu\,\psi_\lambda/\|\psi_\lambda\|_{L^2}$, which satisfies
$E_0[(s^\star)^2]=\mu^2$ and $E_0[\psi_\lambda\,s^\star]=\mu\|\psi_\lambda\|_{L^2}$, establishing \eqref{eq:unrestricted_wcbias}.
Given \eqref{eq:unrestricted_wcbias}, we next minimize $\|\psi_\lambda\|_{L^2}^2$ over $\lambda$:
\[
\|\psi_\lambda\|_{L^2}^2
=E_0[(\phi_c-\lambda\phi_\gamma)^2]
=\sigma_c^2-2\lambda\Sigma_{\gamma c}+\lambda\Sigma_{\gamma\gamma}\lambda',
\]
a strictly convex quadratic since $\Sigma_{\gamma\gamma}\succ0$. The first-order condition is
$-2\Sigma_{\gamma c}+2\Sigma_{\gamma\gamma}\lambda'=0$, i.e.\ $\lambda^\star=\Sigma_{c\gamma}\Sigma_{\gamma\gamma}^{-1}=\Lambda$.
Plugging in yields
\[
\min_\lambda \|\psi_\lambda\|_{L^2}^2
=\sigma_c^2-\Sigma_{c\gamma}\Sigma_{\gamma\gamma}^{-1}\Sigma_{\gamma c}
=\sigma_c^2(1-\mathcal I),
\]
which implies \eqref{eq:unrestricted_minimax_value}. \qed \vspace{1em}

\noindent\textsc{Proof of Corollary~\ref{cor:minimax_mse}:}
For any $\lambda$, the worst-case squared bias follows from Proposition~\ref{prop:minimax}:
\[
\sup_{s\in\mathcal S(\mu)}\big(E_0[\psi_\lambda\,s]\big)^2
= \mu^2\,\|\psi_\lambda\|_{L^2}^2.
\]
The asymptotic variance of $\hat c(\lambda)$ under $P_0$ is $\|\psi_\lambda\|_{L^2}^2$. The minimax squared loss objective is therefore
\[
\sup_{s\in\mathcal S(\mu)}\big(E_0[\psi_\lambda\,s]\big)^2 + \|\psi_\lambda\|_{L^2}^2
= \mu^2\,\|\psi_\lambda\|_{L^2}^2 + \|\psi_\lambda\|_{L^2}^2
= (1+\mu^2)\,\|\psi_\lambda\|_{L^2}^2.
\]
Since $(1+\mu^2)>0$ is a constant that does not depend on $\lambda$, the minimizer over $\lambda$ is the same as the minimizer of $\|\psi_\lambda\|_{L^2}^2$, which by the proof of Proposition~\ref{prop:minimax} is uniquely $\lambda^\star = \Lambda = \Sigma_{c\gamma}\Sigma_{\gamma\gamma}^{-1}$. The minimum value is
\[
(1+\mu^2)\big(\sigma_c^2 - \Sigma_{c\gamma}\Sigma_{\gamma\gamma}^{-1}\Sigma_{\gamma c}\big)
= (1+\mu^2)\,\sigma_c^2(1-\mathcal I). \qed
\]

\section{Definitions, Remarks, and Worked Examples}
\label{sec:app_b}

\noindent\textsc{Worked Example of Selection Distortion:} Let $\hat c_S = \bar Y_1 - \bar Y_0$ and $\hat c_L=\hat c_S-\hat\beta'\hat\gamma$. Suppose for simplicity that $p_\gamma=1$ and, under the base model,
\[
\begin{pmatrix}
Z_S\\
Z_\gamma
\end{pmatrix}
=
\begin{pmatrix}
\sqrt n(\hat c_S-c)/\sigma_S\\
\sqrt n(\hat\gamma-\gamma)/\sigma_\gamma
\end{pmatrix}
\Rightarrow
N\!\left(0,
\begin{bmatrix}
1 & \rho\\
\rho & 1
\end{bmatrix}\right),
\qquad
\rho=\frac{\Sigma_{c_S\gamma}}{\sigma_S\sigma_\gamma}.
\]
Then $Z_S$ admits the decomposition
\[
Z_S=\rho Z_\gamma+\sqrt{1-\rho^2}\,\varepsilon,
\qquad \varepsilon\sim N(0,1),\ \varepsilon\perp Z_\gamma.
\]
Under a ``pass'' rule $S=\mathbf 1\{|Z_\gamma|\le t\}$, the reported short estimate is drawn from the conditional distribution $Z_S\,\big|\,\{|Z_\gamma|\le t\}$, which is a non-Gaussian mixture since the term $\rho Z_\gamma$ is truncated by the event $\{|Z_\gamma|\le t\}$. In particular, its variance differs from the unconditional variance whenever $\rho\neq 0$:
\[
\Var\!\big(Z_S\mid |Z_\gamma|\le t\big)
=(1-\rho^2)+\rho^2\,\Var\!\big(Z_\gamma\mid |Z_\gamma|\le t\big)\neq 1.
\]
A similar point applies to the ``long'' specification $\hat c_L=\hat c_S-\beta_L\hat\gamma$. In standardized form,
\[
Z_L=\frac{\sqrt n(\hat c_L-c)}{\sigma_S}
=Z_S-\kappa_L Z_\gamma
=(\rho-\kappa_L)Z_\gamma+\sqrt{1-\rho^2}\,\varepsilon,
\qquad
\kappa_L=\beta_L\frac{\sigma_\gamma}{\sigma_S},
\]
so $Z_L\,\big|\,\{|Z_\gamma|>t\}$ is likewise distorted unless $\kappa_L=\rho$.
Residualization sets $\kappa_R=\rho$, equivalently $\Lambda=\Sigma_{c_S\gamma}\Sigma_{\gamma\gamma}^{-1}$, yielding
\[
Z_R=Z_S-\rho Z_\gamma=\sqrt{1-\rho^2}\,\varepsilon,
\]
which is independent of $Z_\gamma$. Thus, conditioning on any rule based on $\hat\gamma$ leaves the first-order distribution of the residualized estimator unchanged. \vspace{1em}

\noindent\textsc{Worked Example of Bias Under Local Misspecification:}\label{sec:app_bias_detail} 
We provide here a self-contained derivation of the first-order asymptotic bias formula \eqref{eq:bias_formula_detail}, clarifying why the misspecification score $s$ represents perturbations to the distribution only (not to the parameter of interest), and the condition under which the linear correction does not subtract first-order signal about the target.

\paragraph{Mean shift under local misspecification.}
Under the perturbed distribution $P_n$ defined in \eqref{eq:local_perturb}, the mean of the influence function $\psi_\lambda = \phi_c - \lambda\phi_\gamma$ shifts as follows:
\begin{align}
E_{P_n}[\psi_\lambda(D)]
&= \int \psi_\lambda(d)\left(1 + \frac{s(d)}{\sqrt{n}}\right)dP_0(d) \notag\\
&= E_0[\psi_\lambda] + \frac{1}{\sqrt{n}}E_0[\psi_\lambda \cdot s] 
= \frac{1}{\sqrt{n}}E_0[\psi_\lambda \cdot s], \label{eq:mean_shift_app}
\end{align}
where the last equality uses $E_0[\psi_\lambda]=0$. The mean of the sample average is therefore
\[
E_{P_n}\!\left[\frac{1}{\sqrt{n}}\sum_{i=1}^n \psi_\lambda(D_i)\right]
= \sqrt{n}\cdot E_{P_n}[\psi_\lambda(D)]
= E_0[\psi_\lambda \cdot s].
\]
Le Cam's third lemma \citep[Examples 6.5 and 6.7]{van2000asymptotic} formalizes this as an asymptotic distributional statement: under the drifting sequence $D_i\sim P_n$,
\[
\sqrt{n}\big(\hat c(\lambda) - c\big) \xrightarrow{d} N\!\big(E_0[\psi_\lambda \cdot s],\;\|\psi_\lambda\|_{L^2}^2\big),
\]
so the first-order asymptotic bias is $E_0[\psi_\lambda \cdot s]$.

\paragraph{Decomposition into structural and misspecification scores.}
To clarify why $s$ represents misspecification only, we connect to the framework of \citet{andrews2020informativeness} (hereafter AGS). In AGS, the distribution $F(\eta,\zeta)$ depends on a structural parameter $\eta$ and a misspecification parameter $\zeta$ (with $\zeta=0$ under the base model). Their Lemma~1 considers drifting sequences $\eta = \eta_0 + h/\sqrt{n}$, $\zeta = z/\sqrt{n}$, and shows that the asymptotic mean of the joint estimator is determined by the total score $s_h + s_z$, where
\[
s_h(d) = \frac{\partial}{\partial t_h}\log f_{h,z}(d;0,0), \qquad
s_z(d) = \frac{\partial}{\partial t_z}\log f_{h,z}(d;0,0)
\]
are the scores from perturbing $\eta$ in direction $h$ (which changes $c$) and $\zeta$ in direction $z$ (which changes the distribution but not $c$), respectively.

The minimax problem in Proposition~\ref{prop:minimax} asks: holding the true parameter $c(P_0)$ fixed, how large can the bias of $\hat c(\lambda)$ be? Since perturbations in $\eta$ change $c$, they alter the estimand, not the estimator's accuracy for a fixed target. Formally, under the combined score $s_h + s_z$, the total asymptotic mean of $\sqrt{n}(\hat c(\lambda) - c(\eta_0))$ is $E_0[\psi_\lambda(s_h + s_z)]$, while the true parameter shifts to $c'(h) = E_0[\phi_c\,s_h]$ (the pathwise derivative of $c$ in direction $h$). The bias is therefore
\begin{align}
\text{bias} &= E_0[\psi_\lambda(s_h + s_z)] - c'(h) \notag\\
&= E_0[(\phi_c - \lambda\phi_\gamma)(s_h + s_z)] - E_0[\phi_c\,s_h] \notag\\
&= E_0[\phi_c\,s_z] - \lambda\,E_0[\phi_\gamma\,s_z] - \lambda\,E_0[\phi_\gamma\,s_h] \notag\\
&= E_0[\psi_\lambda\,s_z] - \lambda\,E_0[\phi_\gamma\,s_h]. \label{eq:bias_with_sh_app}
\end{align}
The first term, $E_0[\psi_\lambda\,s_z]$, is the bias from misspecification, identical to \eqref{eq:bias_formula_detail} with $s=s_z$. The second term, $-\lambda\,E_0[\phi_\gamma\,s_h]$, reflects a potential cost of the linear correction: by subtracting $\lambda\hat\gamma$ from $\hat c$, we may inadvertently remove a component that tracks a \textit{real} change in the parameter of interest, introducing bias that is absent from the original estimator $\hat c$.
The second term in \eqref{eq:bias_with_sh_app} vanishes when the structural directions $s_h$ are within-model perturbations and the diagnostic check is identically zero on the model, i.e.\ when $\gamma(P)=0$ for all $P\in\mathcal P$. In that case, the pathwise derivative of $\gamma$ along any such structural direction is zero, so $E_0[\phi_\gamma\,s_h]=0.$ Hence, subtracting $\lambda\hat\gamma$ does not remove first-order signal about the parameter along directions that change $c$, and \eqref{eq:bias_with_sh_app} reduces to $E_0[\psi_\lambda\,s_z]$. The minimax problem over $s=s_z$ alone is therefore the relevant one for fixed-target local misspecification.

This is the same orthogonality that underlies Section~\ref{sec:result2_efficiency}. There it ensures that linear adjustments of the form $\hat c-\lambda\hat\gamma$ preserve the target's pathwise derivative along within-model directions; here it ensures that the only first-order bias to be minimized is the bias coming from misspecification directions that move the data-generating process away from the benchmark model.

\clearpage

\begin{table}[htbp]
\centering
\caption{Baseline Balance}\label{tab:balance}
{\small
{\def\sym#1{\ifmmode^{#1}\else\(^{#1}\)\fi}
\begin{tabular}{@{\extracolsep{0pt}}p{6.5cm}*{3}{>{\centering\arraybackslash}m{2.5cm}}@{}}
\toprule
 & Control Mean & Coef.\ on Cash & P-value \\
 & (1) & (2) & (3) \\
\midrule
\addlinespace[3pt]
Age & 39.188 & -0.447 & 0.596 \\
  & [8.856] & (0.841) & \\ \addlinespace[2pt]
Years of education & 4.721 & -0.060 & 0.857 \\
  & [3.540] & (0.330) & \\ \addlinespace[2pt]
Can read newspaper in Odiya & 0.630 & 0.021 & 0.670 \\
  & [0.484] & (0.048) & \\ \addlinespace[2pt]
Married & 0.984 & -0.012 & 0.468 \\
  & [0.127] & (0.016) & \\ \addlinespace[2pt]
Has any children & 0.891 & -0.038 & 0.262 \\
  & [0.313] & (0.034) & \\ \addlinespace[2pt]
Primarily daily laborer & 0.751 & -0.056 & 0.216 \\
  & [0.433] & (0.045) & \\ \addlinespace[2pt]
Days of paid work in past 7 days & 1.884 & -0.130 & 0.509 \\
  & [2.125] & (0.196) & \\ \addlinespace[2pt]
Days of paid work in past 30 days & 8.602 & 0.098 & 0.889 \\
  & [6.307] & (0.701) & \\ \addlinespace[2pt]
House quality (durable house) & 0.238 & 0.003 & 0.946 \\
  & [0.427] & (0.042) & \\ \addlinespace[2pt]
Owns farmland & 0.568 & 0.012 & 0.788 \\
  & [0.497] & (0.046) & \\ \addlinespace[2pt]
No outstanding food loans & 0.459 & -0.006 & 0.902 \\
  & [0.500] & (0.051) & \\ \addlinespace[2pt]
Can get Rs. 1K in emergency & 0.355 & -0.034 & 0.458 \\
  & [0.480] & (0.046) & \\ \addlinespace[2pt]
Worried about finances & 0.883 & -0.022 & 0.551 \\
  & [0.323] & (0.037) & \\ \addlinespace[2pt]
Worried about any loan & 0.579 & -0.031 & 0.513 \\
  & [0.495] & (0.048) & \\ \addlinespace[2pt]
Amount of loans worried about & 14624.571 & -913.348 & 0.679 \\
  & [15994.347] & (2203.768) & \\ \addlinespace[2pt]
Has loans & 0.683 & 0.027 & 0.550 \\
  & [0.467] & (0.045) & \\ \addlinespace[2pt]
Has moneylender loans & 0.175 & -0.021 & 0.572 \\
  & [0.381] & (0.037) & \\ \addlinespace[2pt]
\bottomrule
\end{tabular}}

}
\end{table}

\begin{table}[htbp]
\centering
\caption{Covariate Contributions to the Residualization Correction}\label{tab:decomposition}
{\small
{\def\sym#1{\ifmmode^{#1}\else\(^{#1}\)\fi}
\begin{tabular}{@{\extracolsep{0pt}}p{6.5cm}*{3}{>{\centering\arraybackslash}m{2.5cm}}@{}}
\toprule
 & $\hat\Lambda_k$ & $\hat\gamma_k$ & $\hat\Lambda_k \hat\gamma_k$ \\
\midrule
\addlinespace[3pt]
Age & -0.0014 & -0.4461 & 0.0006 \\
\addlinespace[2pt]
Years of education & 0.0173 & -0.0596 & -0.0010 \\
\addlinespace[2pt]
Can read newspaper in Odiya & -0.0227 & 0.0205 & -0.0005 \\
\addlinespace[2pt]
Married & 0.0399 & -0.0119 & -0.0005 \\
\addlinespace[2pt]
Has any children & 0.1320 & -0.0378 & -0.0050 \\
\addlinespace[2pt]
Primarily daily laborer & 0.1156 & -0.0559 & -0.0065 \\
\addlinespace[2pt]
Days of paid work in past 7 days & 0.0349 & -0.1295 & -0.0045 \\
\addlinespace[2pt]
Days of paid work in past 30 days & -0.0090 & 0.0975 & -0.0009 \\
\addlinespace[2pt]
House quality (durable house) & -0.0164 & 0.0029 & -0.0000 \\
\addlinespace[2pt]
Owns farmland & -0.0193 & 0.0124 & -0.0002 \\
\addlinespace[2pt]
No outstanding food loans & -0.0561 & -0.0062 & 0.0003 \\
\addlinespace[2pt]
Can get Rs. 1K in emergency & -0.0190 & -0.0342 & 0.0006 \\
\addlinespace[2pt]
Worried about finances & -0.0563 & -0.0219 & 0.0012 \\
\addlinespace[2pt]
Worried about any loan & -0.0834 & -0.0312 & 0.0026 \\
\addlinespace[2pt]
Amount of loans worried about & -0.0000 & -856.2956 & 0.0029 \\
\addlinespace[2pt]
Has loans & 0.0116 & 0.0268 & 0.0003 \\
\addlinespace[2pt]
Has moneylender loans & 0.1242 & -0.0211 & -0.0026 \\
\addlinespace[2pt]
\midrule
Total correction & & & -0.0130 \\
\addlinespace[3pt]
\bottomrule
\end{tabular}}

}
\end{table}

\end{document}